\documentclass[11pt,a4paper]{amsart}
\usepackage[top=4cm,bottom=4cm,left=3.5cm,right=3.5cm]{geometry}

\usepackage{amsmath,amssymb,amsthm}
\usepackage[english]{babel}
\usepackage{mathtools}
\usepackage{enumerate,caption}
\usepackage[numbers]{natbib}
\usepackage[hidelinks]{hyperref}
\usepackage{tikz}
\usepackage{tkz-graph}
\usepackage{tikz-cd}
\tikzset{%
    symbol/.style={%
        ,draw=none
        ,every to/.append style={%
            edge node={node [sloped, allow upside down, auto=false]{$#1$}}}
    }
}
\usepackage{xcolor}
\usepackage{algorithm}
\usepackage[noend]{algpseudocode}

\DeclareMathOperator{\Mat}{Mat}
\DeclareMathOperator{\Cen}{Cen}

\DeclareMathOperator{\End}{End}

\newcommand{\vN}{\mathbb{N}}
\newcommand{\vZ}{\mathbb{Z}}

\newcommand{\epm}{E_p^{(m)}}
\newcommand{\fpm}{F_p^{(m)}}
\newcommand{\zpm}{\vZ/p^{m}\vZ}

\newtheorem{theorem}{Theorem}
\newtheorem{protocol}[theorem]{Protocol}
\newtheorem{lemma}[theorem]{Lemma}
\newtheorem{proposition}[theorem]{Proposition}
\newtheorem{corollary}[theorem]{Corollary}
\theoremstyle{definition}

\newtheorem{definition}[theorem]{Definition}
\theoremstyle{remark}
\newtheorem{remark}[theorem]{Remark}

\usepackage{todonotes}

\numberwithin{equation}{section}

\begin{document}

\title[On the algebraic structure of $\epm$ and applications to cryptography]{On the algebraic structure of $\epm$ and applications to cryptography}

\author[K. Khathuria]{Karan Khathuria}
\address{Institute of Mathematics\\
University of Zurich\\
Winterthurerstrasse 190\\
8057 Zurich, Switzerland\\
}
\email{karan.khathuria@math.uzh.ch}

\author[G. Micheli]{Giacomo Micheli}
\address{Mathematical Institute\\
University of Oxford \\
Woodstock Rd,\\ 
Oxford OX2 6GG, United Kingdom
}
\email{giacomo.micheli@maths.ox.ac.uk}

\author[V. Weger]{Violetta Weger}
\address{Institute of Mathematics\\
University of Zurich\\
Winterthurerstrasse 190\\
8057 Zurich, Switzerland\\
}
\email{violetta.weger@math.uzh.ch}

\thanks{The second author is thankful to Swiss National Science Foundation grant number 171248.}
\subjclass[2010]{}


\begin{abstract}
In this paper we show that the $\zpm$-module structure of the ring $\epm$ is isomorphic to a $\zpm$-submodule of the matrix ring over $\zpm$. Using this intrinsic structure of $\epm$, solving a linear system over $\epm$ becomes computationally equivalent to solving a linear system over $\zpm$. As an application we break the protocol based on the Diffie-Hellman Decomposition problem and ElGamal Decomposition problem over $E_p^{(m)}$. Our algorithm terminates in a provable running time of $O(m^{6})$ $\zpm$-operations. 
\end{abstract}

\maketitle

\section{Introduction}
\label{sec:introduction}
Public key cryptosystems are often based on number theoretical problems, such as integer factorization as in RSA \cite{RSA} or the discrete logarithm problem over finite fields or over elliptic curves. The latter is the base for well known protocols, such as  the ElGamal protocol \cite{elgamal}  or the Diffie-Hellman  key exchange protocol \cite{diffie}. Increasing computing powers threatens these classical cryptographic schemes and new ambient spaces are demanded, for example involving noncommutative structures (see \cite{anshel1999algebraic, ko2007towards, ko2000new, sakalauskas2003basic, sidel1994systems}). 
In nonabelian  groups there are two main problems which give rise to cryptographic schemes; the semigroup action problem (SAP) \cite{maze2005public}, and the decomposition problem (DP). For an overview see \cite{groupcrypto, decomposition}.
These two problems are very similar:  in the SAP one is given a finite semigroup $S$ acting on a finite set $A$, for $x,y \in A$, such that there exists an $s \in S$ with $y=s \star x$, one wants to find $t \in S$, such that $y=t \star x$.
Whereas in the DP one is given a nonabelian group $G$, $(x,y) \in G \times G$ and  $S \subset G$, one wants to find $z_1, z_2 \in S$, such that $y=z_1xz_2$.

Based on these two problems J.J. Climent and J.A. L{\'o}pez-Ramos proposed three protocols in \citep{climent2016public} over a special ring of matrices involving operations modulo different powers of the same prime, called $E_p^{(m)}$. Similar cryptosystems can be found in \citep[Example 4.3.c]{doi:10.1142/S0219498817501481}. This ring is a generalization of the ring $E_p$, Climent, Navarro and Tartosa introduced in \cite{ClimentEp}. 
The first cryptographic scheme based on $E_p$ \cite{climent2012key}, was broken in \cite{Kamal2012}. This attack can be prevented by admitting only few invertible elements, as it is the case in the ring $E_p^{(m)}$  \cite[Corollary 1]{climent2014extension}. In addition, another nice property of such rings is that they do not admit embeddings into matrix rings over a field (see \cite{Bergman1974}), which is often the main problem of cryptographic schemes over matrix rings (see for example \cite{micheli2015cryptanalysis}) and it prevents a reduction to small extensions of finite fields  as in \cite{menezes1997discrete}.

The first protocol proposed in \citep{climent2016public}  by Climent and  L{\'o}pez-Ramos based on the semigroup action problem over the ring $E_p^{(m)}$ was broken by Micheli and Weger  in \cite{micheli2018cryptanalysis} using a solution sieve argument.
The remaining two protocols proposed  in \cite{climent2016public} are based on the decomposition problem over  $E_p^{(m)}$ and happen to be equivalent. They will be denoted by the Diffie-Hellman Decomposition Problem (DHDP) and the ElGamal  Decomposition Problem  (EGDP). 
A cryptanalysis of these two protocols was  considered by Zhang in  \cite{zhang}, where the Cayley-Hamilton Theorem is used to derive a linear system over $\epm$. However, even though the main idea is correct, the system over $\epm$ is then  directly  considered over $\zpm$, where the system not necessarily  admits a solution, as we will show in an example. The running time of the  claimed attack is  $\mathcal{O}(m^7)$ $\zpm$-operations.

In this paper we introduce a new approach for solving the decomposition problem over $\epm$, where we consider a  $\zpm$-module isomorphic to $\epm$. The implication is twofold: on one hand, we fix the issue of the inconsistent system (see also beginning of Section \ref{sec3} for additional details and an example where the attack in \cite{zhang} fails in Subsection \ref{example}), and on the other hand, we speed up the attack. In particular, this new approach solves the decomposition problem in $\mathcal{O}(m^6)$  $\zpm$-operations.


This paper is organized as follows: in Section \ref{sec2} we recall the definitions and properties of the ring  $E_p^{(m)}$ and state the DHDP and EGDP protocols over  $E_p^{(m)}$. In Section \ref{sec3} we present the practical attack on the DHDP protocol, which in turn will also break the equivalent EGDP protocol.
In Subsection \ref{example} we show in an example how the attack works.

\subsection{Notation}
Let $T$ be a subset of a (possibly non-commutative) ring $S$. We will denote the centralizer of $T$ by 
\[\Cen(T)=\{U\in S \ \mid \  UR=RU \ \forall \ R \in T\}.\]
When $T=S$, then $\Cen(S)$ is said to be the center of $S$ and will be denoted by $Z(S)$.
Let $\mathbb{N}$ denote the natural numbers, i.e. $\mathbb{N} =\{1,2, \ldots\}$ and $\mathbb{N}_0 = \mathbb{N} \cup \{0\}$.
For any commutative ring $R$, and any two positive integers $k,m\in \vN$ we will denote by $\Mat_{k\times m}(R)$ the set of $k$ by $m$ matrices with coefficients in $R$. 
If $M$ is an abelian group and  $R$ is a ring acting on $M$, we denote by $\End_R(M)$ the set of endomorphisms of $M$ as an $R$-module. Notice that $\End_R(M)$ has a natural ring structure.
Let $\phi\in \End_R(M)$, we denote by $R[\phi]$ the smallest subring of $\End_R(M)$ which contains $R$ and $\phi$.

\section{The ring $E_p^{(m)}$ and the decomposition problem}\label{sec2}

Let us recall the definition of the matrix ring $E_p^{(m)}$ and its center, which were first introduced in \citep[Theorem 1]{climent2014extension}.

\begin{definition}\label{definitionepm}
Let $E_p^{(m)} $ be the following set of matrices.
\begin{equation*}
E_p^{(m)} = \left\{(a_{ij})_{i,j\in \{1,\dots, m\}} \ \mid \ a_{ij} \in \mathbb{Z}/p^i\mathbb{Z} \ \text{if} \ i \leq j, \ \text{and} \ a_{ij} \in p^{i-j}\mathbb{Z}
/p^i\mathbb{Z} \ \text{if} \ i >j\right\}.
\end{equation*}
To shorten the notation we will write $[a_{ij}]=(a_{ij})_{i,j\in \{1,\dots, m\}}$.
This set forms a ring with the addition and multiplication defined, respectively, as follows
\begin{align*}
[a_{ij}] + [b_{ij}]  & =  \left[(a_{ij}+b_{ij}) \ \text{mod} \ p^i \right], \\
[a_{ij}] \cdot [b_{ij}]  & =  \left[\left( \sum_{k=1}^m a_{ik} b_{kj} \right) \ \text{mod} \ p^i \right].
\end{align*}
\end{definition} 

Let us denote by $V$ the set  $\mathbb{Z}/p\mathbb{Z} \times \cdots \times  \mathbb{Z}/p^m\mathbb{Z}$.
The ring $E_p^{(m)} $ acts on $V$  by the usual matrix multiplication.

\begin{theorem}\citep[Theorem 2]{climent2016public} \label{theorem:center}
The center $Z\left(E_p^{(m)}\right)$ of  $E_p^{(m)} $  is given by the set
\begin{equation*}
\left\{ [a_{ij}] \in E_p^{(m)} \ \bigm| \ a_{ii} = \sum_{j=0}^{i-1} p^j u_j, \ \text{with} \ u_j \in \{0, \ldots, p-1\} \ \text{and} \ a_{ij} = 0 \ \text{if} \ i \neq j \right\}. 
\end{equation*}
\end{theorem}
For $M \in E_p^{(m)}$, let us denote by $\Cen(M)$ the centralizer of $M$, i.e. the set of elements $X \in E_p^{(m)}$, such that $XM =MX$. Define the set \[H(M) = \left \lbrace \sum_{i=0}^k C_iM^i \ \bigm| \ C_i \in Z(\epm), k \in \mathbb{N} \right \rbrace. \] 

Let us recall the Diffie-Hellman decomposition problem, proposed in \citep[Example 3]{decomposition}.

\begin{definition}[DH Decomposition Problem (DHDP)]
Let $G$ be a semigroup, $A, B \subseteq G$ be two subsemigroups such that $ab = ba$ for every $a \in A$ and $b \in B$ and assume that $x \in G$. Given two elements $a_1xa_2$ and $b_1xb_2$, with $a_1, a_2 \in A$ and $b_1, b_2 \in B$, find the element $a_1b_1xb_2a_2$.
\end{definition}

In \citep{climent2016public}, Climent and L{\'o}pez-Ramos proposed two protocols based on the decomposition problem over $\epm$, one of the protocols is a Diffie-Hellman key exchange and the other one is an ElGamal protocol, both analogous to the Diffie-Hellman key exchange \citep{diffie} and the ElGamal cryptosystem \cite{elgamal} respectively.

\begin{protocol}[DHDP protocol] \label{protocol:DHDP} Alice and Bob agree on two public elements $M,X \in E_p^{(m)}$ such that $M \notin \Cen(X)$.
\begin{itemize}
\item[1.] Alice chooses $A_1,A_2 \in H(M)$ and sends $G_A = A_1XA_2$ to Bob.
\item[2.] Bob chooses $B_1,B_2 \in \Cen(M)$ such that $B_1X \neq XB_2$ and sends $G_B = B_1 X B_2 $ to Alice.
\item[3.] Alice computes $A_1 G_BA_2$.
\item[4.] Bob computes $B_1 G_A B_2$.
\end{itemize}
\end{protocol}

Since $A_i$ and $B_i$ commute for all $i \in \{1,2\}$, it is clear that Alice and Bob share a common value.

\begin{protocol}[EGDP protocol]\label{protocol:EGDP} Alice and Bob agree on a public element $M \in E_p^{(m)}$. Let $S \in \epm$ be the secret that Bob wants to send Alice.
\begin{itemize}
\item[1.] Alice chooses $N \in \epm$ such that $NM \neq MN$ and two elements $A_1,A_2 \in H(M)$ and publishes her public key $(N,A_1NA_2)$.
\item[2.] Bob chooses randomly two elements $B_1,B_2 \in \Cen(M)$ and sends $(F,D)=(B_1 N B_2,S + B_1A_1NA_2B_2)$ to Alice.
\item[3.] Alice recovers $S$ by computing $D - A_1 F A_2$.
\end{itemize}
\end{protocol}

Since  $A_i$ and $B_i$ commute for all $i \in \{1,2\}$ we have that $$D - A_1 F A_2 = S+B_1A_1NA_2B_2 - A_1B_1NB_2A_2 =S.$$

As observed in \citep[Theorem 4]{climent2016public}, breaking the EGDP protocol  is equivalent to breaking the DHDP protocol.

\section{Solving the decomposition problem in $\mathcal{O}(m^6)$ $ \zpm$ operations }\label{sec3}
In this section we provide an algorithm to solve the decomposition problem over $\epm$ and therefore to break the DHDP protocol. 
It is worth mentioning that in \citep{zhang} the author claims to have an attack that runs with $\mathcal{O}(m^7)$ $\zpm$-operations. Even though the main idea of the attack is correct (i.e. reducing the problem to solving a system of linear equations over $\epm$), it presents an issue when building the actual linear equations. In fact, the equations have different moduli, depending on which row is considered, so the claim that it is enough to solve the system over $\vZ/p^m\vZ$ is incorrect  (we provide an explicit example where the attack fails in Subsection \ref{example}).  
Moreover, the aforementioned attack would run in $\mathcal{O}(m^7) $ $\zpm$-operations,  instead our attack runs in $\mathcal{O}(m^6)$ $\zpm$-operations, reducing the complexity of the DHDP.
As mentioned in Protocol \ref{protocol:DHDP} and Protocol \ref{protocol:EGDP}, the two subgroups used are $H(M)$ and $\Cen(M)$ for a publicly known $M \in \epm$.
\begin{lemma}
The center of the ring $\epm$ is isomorphic to $\vZ/p^m\vZ$ as rings. \label{lemma:center}
\end{lemma}
\begin{proof}

It is easy to see that the following map is a ring isomorphism

\begin{eqnarray*}
\psi: \mathbb{Z}/ p^m \mathbb{Z} & \to &  Z(E_p^{(m)}) \\
z & \mapsto & [a_{ij}],
\end{eqnarray*}
where $a_{ii} = z \mod p^i$ and $a_{ij}=0$ for $i \neq j$. 
\end{proof}
A direct generalization of Theorem 5 in \cite{climent2011arithmetic} shows that if one considers     $\mathbb{Z}/p\mathbb{Z} \times \cdots \times \mathbb{Z}/p^m\mathbb{Z}$ as a $\mathbb Z$-module, then $E_p^{(m)}$ is isomorphic to  $\End_{\mathbb Z}(\mathbb{Z}/p\mathbb{Z} \times \cdots \times \mathbb{Z}/p^m\mathbb{Z})$. Using this fact and the Cayley-Hamilton Theorem, we can prove that the subring generated by a matrix in $E_p^{(m)}$ is a finite dimensional $\mathbb Z$-module. 
To see this in detail, let us now recall the general statement of the Cayley-Hamilton Theorem.

\begin{theorem}\citep[Proposition 2.4]{atiyah}\label{thm:cayleyhamiltonrings}
Let $R$ be a ring, let $M$ be a finitely generated $R$-module, let $\phi: M \to M$ be a module morphism and let $I \subset R$ be an ideal of $R$, such that $\phi(M) \subseteq IM$. Let $n \in \mathbb{N}$ be the number of elements needed to generate $M$. Then there exist $a_{n-1}, \ldots, a_0 \in I$, such that 
\begin{equation*}
\phi^n + a_{n-1}\phi^{n-1} + \cdots a_0 =0.
\end{equation*}
\end{theorem}
We now prove the corollary we are interested in.
\begin{corollary}
For every $A \in E_p^{(m)}$, there exists $a_0, \ldots, a_{m-1} \in \mathbb{Z}$, such that \label{cor:cayley}
\begin{equation*}
A^m =  a_0+ a_1 A+ \cdots  a_{m-1} A^{m-1}.
\end{equation*} 
\end{corollary}
\begin{proof}
In Theorem \ref{thm:cayleyhamiltonrings}, set $I=R= \mathbb{Z}$ and $M= \mathbb{Z}/p\mathbb{Z} \times \cdots \times \mathbb{Z}/p^m\mathbb{Z}$, hence $n=m$ and $\phi$ is a matrix in $\epm$. It follows now immediately that $\mathbb Z[\phi]$ has dimension less than or equal to $m$ (as a $\mathbb Z$-module). 
\end{proof}
\begin{remark}
Notice that in the statement and the proof of Corollary \ref{cor:cayley}, $\mathbb Z$ could as well be replaced by $\mathbb Z/p^m \mathbb Z$ since any element in $p^m\mathbb Z$ acts as the zero morphism over $M$.
\end{remark}

\begin{lemma}
Let $M \in \epm$. Then the map $\psi : \left ( \vZ/p^m\vZ \right )[x] \rightarrow  H(M)$ given by $\psi(f(x)) = f(M)$ is a surjective $\vZ/p^m\vZ$-algebra homomorphism. \label{lem:Hm}
\end{lemma}
\begin{proof}
First, using Lemma \ref{lemma:center}  one can identify the center of $\epm$ with $\vZ/p^m\vZ$, from which  follows that the map is well defined.
To see that $\psi$ is a surjective homomorphism, it is enough to consider $\epm$ as $\End_{\mathbb Z/p^m \mathbb Z}(\mathbb{Z}/p\mathbb{Z} \times \cdots \times \mathbb{Z}/p^m\mathbb{Z})$ and to notice that $H(M)\cong (\vZ/p^m\vZ)[M]$:  in fact, using Lemma \ref{lemma:center}, there exist $u_0,\ldots,u_{k} \in \vZ/p^m\vZ$ such that each $C_i$  of a matrix in $H(M)$ can be written as the diagonal matrix with entries $(u_i \mod p, u_i \mod p^2, \ldots, u_i \mod p^m)$. 

%
\end{proof}


\begin{proposition}
Let $M,X \in \epm $  and $G_A = A_1 X A_2$ for some $A_1,A_2 \in H(M)$. Then there exists $\lambda_{11}, \lambda_{12}, \ldots ,\lambda_{mm} \in \vZ/p^m\vZ$ such that $G_A = \sum_{i,j = 0 }^{m-1} \lambda_{ij} M^i X M^j $. \label{proposition:break}
\end{proposition}

\begin{proof}
Combining Lemma \ref{lem:Hm} and Corollary \ref{cor:cayley}, we can write $A_1 = \sum_{i=0}^{m-1} u_i M^ i$ and $A_2 = \sum_{i=0}^{m-1} v_i M^i$ for some  $u_0,\ldots,u_{m-1},v_0,\ldots,v_{m-1} \in \zpm$. Then \begin{eqnarray*}
G_A & = & A_1 X A_2 \\
 & = & \left ( \sum_{i=0}^{m-1} u_i M^i  \right ) X \left ( \sum_{j=0}^{m-1} v_j M^j \right ) \\
 & = & \sum_{i=0}^{m-1} \sum_{j=0}^{m-1} u_i v_j M^i X M^j \\
 & = & \sum_{i=0}^{m-1} \sum_{j=0}^{m-1} \lambda_{ij} M^i X M^j, 
\end{eqnarray*} for $ \lambda_{ij} = u_iv_j$.
\end{proof}

%
%
%
\subsection{The pseudo-$\epm$ Approach}
In this subsection we solve the linear system given in Proposition \ref{proposition:break} by defining the following $\zpm$-modules:
\begin{definition}
Let $\fpm \subseteq \Mat_{m \times m}(\zpm)$ be the following set of matrices.
\[\fpm = \left \lbrace (a_{ij})_{i,j \in \lbrace 1,\ldots, m \rbrace} \mid a_{ij} \in p^{\ell}\zpm \mbox{ where } \ell = \max(m-i,m-j) \right \rbrace. \]
\end{definition}
It is easy to check that $\fpm$ is a $\zpm$-submodule of $\Mat_{m \times m}(\zpm)$. Moreover the following proposition shows that $\fpm$ is isomorphic to $\epm$ as $\zpm$-modules, where the scalar multiplication in $\epm$ is as follows: for $r \in \zpm$ and $[a_{i,j}] \in \epm$  we have that
\begin{align*}
r[a_{i,j}] = [ra_{i,j}] = \left( r a_{i,j} \mod p^i \right)_{i,j \in \{1, \ldots, m\}}.
\end{align*}
\begin{proposition}
$\epm$ is isomorphic to $\fpm$ as $\zpm$-modules. \label{proposition:pseudo_epm}
\end{proposition}
\begin{proof} One can easily check that the following map is an isomorphism
between $\epm$ and $\fpm$
\begin{eqnarray*}
\delta: \quad \epm & \to &  \fpm \\
\left ( a_{ij} \right)_{i,j\in \{1,\dots, m\}} & \mapsto & \left( a_{ij} p^{m-i} \right)_{i,j\in \{1,\dots, m\}}.
\end{eqnarray*}


\end{proof}
\begin{corollary}
Let $M_1,M_2,\ldots,M_k,G \in \epm$. Then $\lambda = (\lambda_1,\lambda_2,\ldots,\lambda_k) \in \left (\zpm\right )^k$ is a solution for $\sum_{i=1}^k \lambda_i M_i = G$ if and only if $\sum_{i=1}^k \lambda_i \delta(M_i) = \delta(G)$. \label{cor:pseudo_epm}
\end{corollary}
\begin{proof}
The proof follows directly by applying the isomorphism $\delta$ to both sides of the equation $\sum_{i=1}^k \lambda_i M_i = G$.
\end{proof}

Now we are ready for the main result.

\begin{theorem}\label{break_DHDP}
DHDP protocol over $\epm$ can be broken in polynomial time.
\end{theorem}

\begin{proof}

Let $M,X \in \epm$ such that $MX \neq XM$, and let $A_1,A_2 \in H(M)$ and $B_1,B_2 \in \Cen(M)$. Given $M,X, G_A = A_1 X A_2$ and $G_B = B_1 X B_2$, we have to find $A_1G_BA_2$.

Using Proposition \ref{proposition:break}, we know that there exist $\lambda_{11}, \ldots ,\lambda_{mm} \in \vZ/p^m\vZ$ such that $G_A = \sum_{i,j = 0 }^{m-1} \lambda_{ij} M^i X M^j $. We use Corollary \ref{cor:pseudo_epm} to solve this system of linear equations for $\lambda_{11},\lambda_{12} \ldots ,\lambda_{mm}$.  Then the exchanged secret is given by
\begin{eqnarray*}
\sum_{i,j=0}^{m-1} \lambda_{ij} M^i G_B M^j & = & \sum_{i,j=0}^{m-1} \lambda_{ij} M^i B_1 X B_2 M^j \\
& = &   \sum_{i,j=0}^{m-1} \lambda_{ij} B_1 M^i X M^j  B_2 \\
& = &   B_1 \left ( \sum_{i,j=0}^{m-1} \lambda_{ij} M^i X M^j \right ) B_2 \\
& = &  A_1 G_B A_2  =  B_1 G_A B_2.  
\end{eqnarray*} 

Algorithm \ref{break} provides a formal way to solve the DHDP protocol over $\epm$. 

\end{proof}

\begin{algorithm}
\caption{Break protocol based on DHDP over $E_p^{(m)}$ using pseudo-$\epm$}\label{break}
\begin{flushleft}
Input: $M,X,G_A,G_B \in E_p^{(m)}$\\ 
Output: the exchanged secret $A_1G_BA_2 \in E_p^{(m)}$
\end{flushleft}
\begin{algorithmic}[1]
\State Construct the matrix of linear equations arising from $A_1XA_2 = G_A$ using Proposition \ref{proposition:break}, given by 
\begin{equation*} G_A = \sum_{i,j =0}^{m-1} \lambda_{ij} M^iXM^j,
\end{equation*} where $\lambda_{ij}$'s are unknown.
\State Apply the $\zpm$-module isomorphism $\delta$ mentioned in Lemma \ref{proposition:pseudo_epm} to the above equation \begin{equation*} \delta(G_A) = \sum_{i,j =0}^{m-1} \lambda_{ij} \delta(M^iXM^j).
\end{equation*}
\State Solve the system of $m^2$ linear equations in $m^2$ unknowns over $\zpm$, generated by equating entries of above matrix equality. By Proposition \ref{proposition:break} and Corollary \ref{cor:pseudo_epm}, such a solution exists.
\State Return $\sum_{i,j =0}^{m-1}\lambda_{ij} M^iG_BM^j$.
\end{algorithmic}
\end{algorithm}

\textbf{Running time.} The running time of Algorithm \ref{break} is given by solving $m^2$ linear equations in $m^2$ unkowns over $\zpm$, which costs  $\mathcal{O}((m^2)^3)$ $\mathbb Z/p^m\mathbb Z$-operations, or  $\mathcal{O}(m^8 \log(p)^2)$ bit operations. 
In \cite{climent2014extension}, Climent et. al. proposed to use the DHDP protocol and EGDP protocol for the parameters $p=2$ and $m=128$. In our implementation, Algorithm \ref{break} took 23.1 days to break these parameters. The results were obtained by a MAGMA \cite{MR1484478} implementation using a personal computer with processor Intel Core 6C i7-8700K at 3.7 GHz and 64 GB RAM (see \url{www.math.uzh.ch/aa/uploads/media/attack_CLR.txt}). 
%

\subsection{A  $ 2 \times 2$ example}\label{example}
In the following we provide an example, which serves two purposes, first it shows the Algorithm \ref{break} in practice and second it provides an example where the claimed attack in \cite{zhang} does not work.

Let $m=2, p=5$
and let $M = \begin{bmatrix}
4 & 3 \\15 & 20
\end{bmatrix}$ and $ \ X = \begin{bmatrix}
 0 & 4  \\ 15 & 4
\end{bmatrix}$ be public elements.

Alice chooses \[A_1 = \begin{bmatrix}
1 & 3 \\ 15& 17
\end{bmatrix} \quad A_2 = \begin{bmatrix}
0 & 3 \\15& 11
\end{bmatrix}\] and publishes $G_A= \begin{bmatrix}
0 & 1  \\ 20& 23
\end{bmatrix}$. Bob chooses \[B_1 = \begin{bmatrix}
3 & 3 \\ 15& 9 
\end{bmatrix} \quad B_2 = \begin{bmatrix}
3 & 0 \\ 0 & 18
\end{bmatrix}\] and publishes $G_B= \begin{bmatrix}
0 & 2 \\ 5 & 3
\end{bmatrix}$. The shared secret is then \[A_1G_BA_2 = B_1G_AB_2 = \begin{bmatrix}
 0 & 1 \\ 15& 21
 \end{bmatrix}.\]

The attacker sees only $M,X,G_A,G_B$
and wants to find $A_1 G_B A_2 \in E_5^{(2)}$. \\ In Step 1 of Algorithm \ref{break}, the attacker constructs

\begin{align}\label{examplesys} G_A &= \sum_{i,j =0}^{m-1} \lambda_{ij} M^iXM^j \nonumber \\ & = \lambda_{00} \begin{bmatrix} 0 & 4 \\15 & 4 \end{bmatrix}
 +\lambda_{01} \begin{bmatrix} 10& 5 \\ 20& 0 \end{bmatrix} + \lambda_{10} \begin{bmatrix} 20& 3 \\ 0 & 15\end{bmatrix} + \lambda_{11} \begin{bmatrix} 0 & 20\\ 0 & 0 \end{bmatrix}.
\end{align}
In the second step the attacker applies $\delta$ getting

\begin{equation*} \delta(G_A)= \begin{bmatrix} 0 & 5 \\ 20& 23\end{bmatrix}  =  \lambda_{00} \begin{bmatrix} 0 & 20 \\15 & 4 \end{bmatrix}
 +\lambda_{01} \begin{bmatrix} 0 & 0 \\ 20& 0 \end{bmatrix} + \lambda_{10} \begin{bmatrix} 0 & 15 \\ 0 & 15\end{bmatrix} + \lambda_{11} \begin{bmatrix} 0 & 0\\ 0 & 0 \end{bmatrix}.
\end{equation*}

From this we get the system of $m^2$ linear equations in $m^2$ unknowns $\lambda_{ij}$, but since we applied $\delta$ these are now equations over $\zpm$:
\begin{align*} 
0 \lambda_{00} + 0 \lambda_{01} + 0 \lambda_{10} +0 \lambda_{11}& =0, \\
20 \lambda_{00} + 0 \lambda_{01} + 15\lambda_{10} + 0 \lambda_{11}& =5, \\
15 \lambda_{00} + 20\lambda_{01} + 0 \lambda_{10} + 0 \lambda_{11}& =20,\\
4 \lambda_{00} + 0 \lambda_{01} + 15\lambda_{10} + 0 \lambda_{11}& =23.
\end{align*}
%

One particular solution of this system over $\zpm$ is given by 
$$ \left(\lambda_{00}, \lambda_{01}, \lambda_{10}, \lambda_{11}\right) = \left(
2, 22, 1, 0
\right). $$
The attacker now computes

\begin{align*}
 \sum_{i,j=0}^{m-1} \lambda_{ij} M^i G_B M^j & = 2 G_B + 22 G_B M + 1 M G_B + 0 MG_B M \\ 
 &=\begin{bmatrix}
0 & 1 \\ 15 & 21 
\end{bmatrix}  = A_1G_B A_2.
\end{align*}

Compared to our attack, the approach presented in \cite{zhang} does not make use of Lemma \ref{lemma:center}, Lemma \ref{lem:Hm} and Corollary \ref{cor:pseudo_epm}. Instead the elements of $Z(\epm)$ are seen as diagonal matrices having entries $(a_0,a_0+pa_1,\ldots,a_0 + p a_1 + \cdots + p^{m-1}a_{m-1})$ where $0 \leq a_0,a_1,\ldots,a_{m-1} \leq p-1$. Using this representation and Cayley-Hamilton theorem results in a linear system over $\epm$ of $m^2$ equations in $m^3$ unknowns, as mentioned in Theorem 2 of \cite{zhang}. 
One should observe that this system does not necessarily admit a solution over $\zpm$, which is the approach used in \cite{zhang}. The above mentioned example provides an instance where this approach fails.

Using the approach in \cite{zhang}, the analogue of equation \eqref{examplesys}  is 
\begin{align*}\label{examplesys} G_A &= \sum_{i,j =0}^{m-1} W_{ij} M^iXM^j, \nonumber 
\end{align*}
where $W_{ij} = \begin{bmatrix}
a_0^{ij} & 0 \\ 0 & a_0^{ij} + 5 a_1^{ij}
\end{bmatrix} \in Z(\epm)$. This results in the following system of linear equations:
\begin{align*} 
0 a_0^{00} + 0 a_1^{00} + 0 a_0^{01} +0 a_1^{01} +  0 a_0^{10}  + 0 a_1^{10} + 0 a_0^{11} + 0 a_1^{11}& = 0  \mod 5, \\
4 a_0^{00} + 0 a_1^{00} + 0 a_0^{01} + 0 a_1^{01} + 15 a_0^{10}  + 0 a_1^{10} + 0 a_0^{11} + 0 a_1^{11}& = 5 \mod  5,\\
15 a_0^{00} + 0 a_1^{00} + 20 a_0^{01} + 0 a_1^{01} +  0 a_0^{10}  + 0 a_1^{10} + 0 a_0^{11} + 0 a_1^{11}& = 20 \mod 25, \\
4 a_0^{00} + 20 a_1^{00} + 0 a_0^{01} +0 a_1^{01} +  15 a_0^{10}  + 5 a_1^{10} + 0 a_0^{11} + 0 a_1^{11}& = 23 \mod 25.\\
\end{align*}
In Section 4 of \cite{zhang}, the author claims that it is enough to consider this system over $\vZ/25 \vZ$. However in this example the claim does not hold and shows that the approach used in \cite{zhang} to solve a linear system over $\epm$ is incorrect.

\section*{Acknowledgment}
We would like to thank Alessandro Neri for fruitful discussions and Juan-Antonio  L{\'o}pez-Ramos and Joan-Josep Climent for bringing \cite{zhang}  to our knowledge. The authors would also like to thank the anonymous reviewers for their helpful comments.
The second author is thankful to the Swiss National Science Foundation under grant number 171249. This work has also been supported partly by the Swiss National Science Foundation under grant number 169510.

\bibliographystyle{unsrt}
\bibliography{biblio}

\end{document}